\documentclass[runningheads]{llncs}
\usepackage{graphicx}

\usepackage{amssymb}			
\usepackage{amsmath}	

\usepackage{float}
\usepackage{mathrsfs}
\usepackage[margin=1.2in]{geometry}
\usepackage{array}								
\usepackage{xspace}
\usepackage{color}
\usepackage[utf8]{inputenc}
\usepackage{psfrag}

\usepackage[toc,page]{appendix}

\usepackage[normalem]{ulem}
\usepackage{soul}

\usepackage{array}
\usepackage{booktabs}
\usepackage{multirow}
\usepackage{hhline}
\usepackage{graphicx}

\newcommand*{\QEDB}{\hfill\ensuremath{\square}}%

\newcommand{\NE}{\textsc{ne}\xspace}

\newcommand{\PoA}{\text{PoA}\xspace}
\newcommand{\remove}[1]{}

\title{On the Price of Anarchy for High-Price Links}

\usepackage{microtype}

\author{C. \`Alvarez \and A. Messegu\'e}
\institute{ALBCOM Research Group, Computer Science Department, UPC, Barcelona\\
\email{\{alvarez,amessegue\}@cs.upc.edu}}

\begin{document}

\maketitle

\begin{abstract}

We study Nash equilibria and the price of anarchy in the classic model of Network Creation Games introduced by Fabrikant,  Luthra, Maneva, Papadimitriou and Shenker in 2003. This is a selfish network creation model where players correspond to nodes in a network  and each of them can create links to the other $n-1$ players at a prefixed price $\alpha > 0$. The player's goal is to minimise the sum of her cost buying edges and her cost for using the resulting network. One of the main conjectures for this model states that the price of anarchy, i.e. the relative cost of the lack of coordination,  is constant for all  $\alpha$. This conjecture has been confirmed for $\alpha = O(n^{1-\delta})$ with $\delta \geq 1/\log n$ and for $\alpha > 4n-13$. The best known upper bound on the price of anarchy for the remaining range is $2^{O(\sqrt{\log n})}$. 

We give new insights into the structure of the Nash equilibria for $\alpha > n$ and we enlarge the range of the parameter $\alpha$ for which the price of anarchy is constant. Specifically, we prove that for any small  $\epsilon>0$,    the price of anarchy is constant  for $\alpha > n(1+\epsilon)$ by showing that any biconnected component of any non-trivial Nash equilibrium, if it exists, has at most a constant number of nodes. 
 \end{abstract}

\section{Introduction}Many distinct network creation models trying to capture properties of Internet-like networks or social networks have been extensively studied in Computer Science, Economics, and Social Sciences. In these models, the players (also called nodes or agents) buy some links to other players creating in this way a network formed by their choices. Each player has a cost function that captures the need of buying few links and, at the same time, being well-connected to all the remaining nodes of the resulting network. The aim of each player is  to minimise her cost following her selfish interests.
A stable configuration  in which every player or agent has no incentive in deviating unilaterally from her current strategy is called a \emph {Nash equilibrium} (\NE).
In order to evaluate the social impact of the resulting network, the \emph {social cost} is introduced. In this setting the social cost is defined as the sum of the individual costs of all the players. Since there is no coordination among the different players, one can expect that stable networks do not minimise the social cost. 
The \emph {price of anarchy} (\PoA) is a measure that quantifies how far  is the worst  \NE (in the sense of social cost) with respect to any optimal configuration that minimises the social cost. Specifically, the \PoA is  defined as the ratio between the maximum social cost of \NE and the social cost  of the optimal configuration.
If we were able to prove formally that the \PoA is constant, then we could conclude that the equilibrium configurations in the selfish network creation games are so good in terms of social cost.

Since the introduction of the classical network creation game by Fabrikant et al. in \cite{Fe:03},  many efforts have been done in order to analyze the quality of the resulting equilibrium networks.
The \emph{constant \PoA conjecture} is a well-known conjecture that states  that the \PoA is constant independently of the price of the links. 
In this work we provide a new understanding of the structure of the equilibrium networks for  the classical network creation game \cite{Fe:03}. We focus on the equilibria for high-price links and show that in the case that an equilibrium is not a tree, then the size of any of its biconnected components is upper bounded by a constant. This is the key ingredient to prove later that,  for any small 
$\epsilon>0$, the \PoA  is constant for  
$\alpha > n(1+\epsilon)$  where $\alpha$ is the price per link and $n$ is the number of nodes.

Let us first define formally  the model and  related concepts.


\vskip 10pt

 \subsection{Model and definitions} 
 The \emph{sum classic network creation game} $\Gamma$ is defined by  a pair $\Gamma= (V,\alpha)$ where $V = \left\{1,2,....,n \right\}$  denotes the set of players and  $\alpha>0$ a positive parameter. Each player $u\in V$ represents a node of an undirected graph and $\alpha$ represents the cost of establishing a link.
 
 A \emph{strategy} of a player $u$ of $\Gamma$  is a subset $s_u \subseteq V \setminus \left\{u \right\}$, the set of nodes for which player $u$ pays for establishing a link. A strategy profile for $\Gamma$ is a tuple $s=(s_1,\ldots,s_n)$ where $s_u$ is the strategy of player $u$, for each player $u\in V$. Let $\mathscr{S}$ be the set of all strategy profiles of $\Gamma$. Every strategy profile $s$ has associated a \emph{communication network} that is defined as the undirected graph $G[s] = (V, \left\{uv \mid v \in s_u \lor u \in s_v \right\})$. Notice that $uv$ denotes the undirected edge between $u$ and $v$. 
 
 
Let $d_G(u,v)$ be the distance in $G$ between $u$ and $v$. The cost associated to a player $u \in V$  in a strategy profile $s$ is defined by $c_u(s) =  \alpha |s_u| + D_{G[s]}(u)$ where   $D_G(u) = \sum_{v\in V, v \neq u} d_{G}(u,v)$ is the sum of the distances from the player $u$ to all the other players in $G$. As usual, the social cost of a strategy profile $s$ is defined by  $C(s)= \sum_{u \in V}{c_u(s)}$. 
 
 A Nash Equilibrium (\NE) is a strategy vector $s$ such that  for every player $u$ and every strategy vector $s'$ differing from $s$ in only the $u$ component, $s_u \neq s_u'$,  satisfies $c_u(s) \leq c_u(s')$. In a \NE $s$ no player has incentive to deviate individually her strategy since the cost difference $c_u(s')-c_u(s) \geq 0$. 
  Finally, let us denote by $\mathcal{E}$ the set of all \NE strategy profiles. The  price of anarchy (\PoA ) of $\Gamma$  is defined as $PoA= \max_{s \in \mathcal{E}} 
 C(s)/\min_{s \in \mathscr{S}}  C(s)$. 

It is worth observing that in  a \NE $s=(s_1,...,s_n)$ it never happens that $u\in s_v$ and $v\in s_u$, for any $u,v\in V$. Thus, if $s$ is a \NE, $s$ can be seen as an orientation of the edges of $G[s]$ where an arc from $u$ to $v$ is placed whenever $v \in s_u$. It is clear that a \NE $s$ induces a graph $G[s]$ that we call  \emph{NE graph} and we mostly omit the reference to such strategy profile $s$ when it is clear from context. However, notice that a graph $G$ can have different orientations. Hence, when we say that $G$ is a \NE graph we mean that $G$ is the outcome of a \NE strategy profile $s$, that is, $G=G[s]$. 

Given a graph $G$ we denote by $X \subseteq G$ the subgraph of $G$ induced by $V(X)$. In this way, given a graph  $G=G[s] = (V,E)$, a node $v \in V$,  and $X \subseteq G$, the \emph{outdegree of} $v$ in $X$ is defined as $deg_X^+(v)= | \left\{ u \in V(X) \mid u \in s_v\right\}| $, the \emph{indegree of} $v$ in $X$  as $deg_X^-(v) =  | \left\{ u \in V(X) \mid  v \in s_u \right\}|$, and, finally, the \emph{degree of} $v$ in $X$ as $deg_X(v) = deg_X^+(v)+deg_X^-(v)$. Notice that $deg_X(v) = |\left\{ u \in V(X) \mid  uv \in E \right\}|$.  Furthermore, the average degree of $X$ is defined as $deg(X)= \sum_{v \in V(X)}deg_X(v)/|V(X)|$.

Furthermore, remind that in  a connected graph $G=(V,E)$ a  vertex is a \emph{cut vertex} if its removal increases the number of connected components of $G$. A graph is biconnected if it has no cut vertices. We say that $H \subseteq G$ is a \emph{biconnected component} of $G$  if $H$ is a maximal biconnected subgraph of $G$. More specifically, $H$ is such that there is no other distinct biconnected subgraph of $G$ containing $H$ as a subgraph. Given a biconnected component $H$ of  $G$ and a node  $u \in V(H)$, we define $S(u)$ as the connected component containing $u$ in  the subgraph induced by the vertices $(V(G)\setminus V(H)) \cup \left\{u \right\}$. The \emph{weight} of a node $u \in V(H)$, denoted by $|S(u)|$ is then defined as the number of nodes of $S(u)$. Notice that $S(u)$ denotes the set of all nodes $v$ in the connected component containing $u$ induced by $(V(G) \setminus V(H)) \cup \{u\}$ and then, every shortest path in $G$ from $v$ to any node $w \in V(H)$ goes through $u$. 

In the following sections we consider $G$ to be a \NE  for $\alpha > n $  and $H \subseteq G$, if it exists, a non-trivial biconnected component of $G$, that is, a biconnected component of $G$ of at least three distinct nodes. Then we use the abbreviations $d_G,d_H$ to refer to the diameter of $G$ and the diameter of $H$, respectively, and $n_H$ the size of $H$.

\subsection{Historical overview}

We now describe the progress around the central question of giving improved upper bounds on the \PoA of the network creation games introduced by Fabrikant et al. in \cite{Fe:03}. 

First of all, let us explain briefly two key results that are used to obtain better upper bounds on the \PoA. The first is that the \PoA for trees is at most $5$ (\cite{Fe:03}). The second one is that the \PoA of any \NE graph is upper bounded by its diameter plus one unit (\cite{Demaineetal:07}). Using these two results it can be shown that the \PoA is constant for almost all values of the parameter $\alpha$. Demaine et al. in \cite{Demaineetal:07} showed constant \PoA for  $\alpha = O(n^{1-\delta})$ with $\delta \geq \frac{1}{\log n}$ by proving that the diameter of equilibria is constant for the same range of $\alpha$. In the view that the \PoA is constant for a such a wide range of values of $\alpha$, Demaine et al. in \cite{Demaineetal:07} conjectured that the \PoA is constant for any $\alpha$. This is what we call the \emph{constant \PoA conjecture}. More recently, Bil{\`{o}} and Lenzner in \cite{Lenznertree} demonstrated constant \PoA for $\alpha > 4n-13$ by showing that every \NE is a tree for the same range of $\alpha$. For the remaining range Demaine et al. in \cite{Demaineetal:07} determined that the \PoA is upper bounded by $2^{O(\sqrt{\log n})}$. 

The other important conjecture, the \emph{tree conjecture}, stated by Fabrikant et al. in \cite{Fe:03}, still remains to be solved. The first version of the tree conjecture said that there exists a positive constant $A$ such that every \NE is a tree for $\alpha > A$. This was later refuted by Albers et al. in \cite{Albersetal:06}. The reformulated tree conjecture that is believed to be true is for the range $\alpha > n$. In \cite{Mihalaktree} the authors show an example of a non-tree \NE for the range $\alpha = n-3$ and then, we can deduce that the generalisation of the tree conjecture for $\alpha > n$ cannot be extended to the range $\alpha > n(1-\delta)$ with $\delta >0$ any small enough positive constant. Notice that the constant \PoA conjecture and the tree conjecture are related in the sense that if the tree conjecture was true, then we would obtain that the \PoA is constant for the range $\alpha > n$ as well. 

Let us describe the progress around these two big conjectures considering first the case of large values of $\alpha$ and after the case of small values of $\alpha$. 

\vskip 5pt

For \emph{large values} of $\alpha$ it has been shown constant \PoA for the intervals $\alpha > n^{3/2}$ \cite{Lin}, $\alpha > 12n \log n$ \cite{Albersetal:06}, $\alpha > 273n$ \cite{Mihalakmostly}, $\alpha > 65n$ \cite{Mihalaktree}, $\alpha > 17n$ \cite{Alvarezetal} and $\alpha > 4n-13$ \cite{Lenznertree}, by proving that every \NE for each of these ranges is a tree, that is, proving that the tree conjecture holds for the corresponding range of $\alpha$. 

The main approach to prove the result in \cite{Mihalakmostly,Mihalaktree,Alvarezetal} is to consider a biconnected (or $2$-edge-connected in \cite{Alvarezetal})  component $H$ from the \NE network, and then to establish non-trivial upper and lower bounds for the average degree of $H$, noted as $deg(H)$. More specifically, it is shown that $deg(H) \leq f_1(n,\alpha)$ for every $\alpha \geq c_1 n$ and $deg(H) \geq f_2(n,\alpha)$ for every $\alpha \geq c_2 n$, with $c_1,c_2$ constants  and $f_1(n,\alpha),f_2(n,\alpha)$ functions of $n,\alpha$.  From this it can be concluded that there cannot exist any biconnected component $H$ for any $\alpha$ in the set $ \left\{ \alpha  \mid f_1(n,\alpha) < f_2(n,\alpha) \land \alpha \geq \max(c_1,c_2) n \right\} $, and thus every \NE is a tree for this range of $\alpha$.  

In \cite{Mihalakmostly,Mihalaktree}, to prove the upper bound on the term $deg(H)$ the authors basically consider a BFS tree $T$ rooted at a node $u$ minimising the sum of distances in $H$ and define a \emph{shopping vertex} as a vertex from $H$ that has bought at least one edge of $H$ but not of $T$. The authors show that every shopping vertex has bought at most one extra edge and that the distance between two distinct shopping vertices is lower bounded by a non-trivial quantity that depends on $\alpha$ and $n$. By combining these two properties the authors can give an improved upper bound on $deg(H)$ which is close to $2$ from above when $\alpha$ is large enough in comparison to $n$.  On the other hand, to prove a lower bound on $deg(H)$ the authors show that in $H$  there cannot exist too many nodes of degree $2$ close together. 

In \cite{Alvarezetal},  the authors use the same upper bound as the one in \cite{Mihalaktree} for the term $deg(H)$ but give an improved lower bound better than the one from \cite{Mihalaktree}. To show this lower bound we introduce the concept of \emph{coordinates} and \emph{2-paths}. For $\alpha > 4n$, the authors prove that every minimal cycle is directed and then use this result to show that there cannot exist long $2-$paths.  

In contrast, Bil{\`{o}} and Lenzner in \cite{Lenznertree} consider a different approach. Instead of using the technique of bounding the average degree, they introduce, for any non-trivial biconnected component $H$ of a graph $G$,  the concepts of \emph{critical pair}, \emph{strong critical pair}, and then, show that every minimal cycle for the corresponding range of $\alpha$ is directed. The authors play with these concepts in a clever way in order to reach the conclusion. 

\vskip 3pt

In a very preliminary draft \cite{Alvarezetal2}, we take another perspective and conclude that given $\epsilon > 0$ any positive constant, the \PoA is constant for $\alpha > n(1+\epsilon)$. Specifically, in \cite{Alvarezetal2}, we prove that if the diameter of a \NE graph is larger than a given positive constant, then the graph must be a tree. Such proposal  represents an interesting approach to the same problem but the calculations and the proofs are very involved and hard to follow. In this work we present in a clear and elegant way  the stronger result that, for the same range of $\alpha$, the size of any biconnected component of any non-tree \NE is upper bounded by a constant. 
\vskip 5pt

 For \emph{small values} of $\alpha$, among the most relevant results, it has been proven that the \PoA is  constant for the intervals $\alpha = O(1)$ \cite{Fe:03}, $\alpha = O(\sqrt{n})$ \cite{Albersetal:06,Lin} and $\alpha = O(n^{1-\delta})$ with $\delta \geq 1/\log n$ \cite{Demaineetal:07}. 

The most powerful technique used in these papers is the one from Demaine et al. in \cite{Demaineetal:07}. They show that the \PoA is constant for $\alpha =O(n^{1-\delta})$ with $\delta > 1/\log n$, by studying a specific setting where some disjoint balls of fixed radius are included inside a ball of bigger radius. Considering the deviation that consists in buying the links to the centers of the smaller balls, the player performing such deviation gets closer to a majority of the nodes by using these extra bought edges (if these balls are chosen adequately). With this approach it can be shown that the size of the balls grows in a very specific way, from which then it can be derived the upper bounds for the diameter of equilibria and thus for the \PoA. 

\subsection{Our contribution}

Let us consider a weaker version of the tree conjecture that considers the existence of  biconnected components in a \NE having some specific properties regarding  their size. 

\begin{conjecture}[The biconnected component conjecture]
For $\alpha > n$, any biconnected component of a non-tree \NE graph has size at most a prefixed constant.
\end{conjecture}

 Let $\epsilon > 0$ be any positive constant. We show that the restricted version of this conjecture where $\alpha >n(1+\epsilon)$ is true  (Section \ref{sec:last}, Theorem \ref{thm:weakertreeconjecture}). This result jointly with $d_G \leq d_H+250$ (Theorem \ref{corol:diameter}, Section \ref{sec:diamH2}) for $\alpha > n$,  whenever $H$ exists, imply that $d_G$ is upper bounded by a prefixed constant, too. Recall that, the diameter of any graph plus one unit is an upper bound on the \PoA and the price of anarchy for trees is constant. Hence, we can conclude that the \PoA is constant for $\alpha > n(1+\epsilon)$. 


In order to show these results, we introduce a new kind of sets,  the $A$ sets, satisfying some interesting properties and we adapt some well-known techniques and then, combine them together in a very original way. Let us describe the main ideas of our approach:

\begin{itemize}

\item  Inspired by the technique considered in \cite{Demaineetal:07} which is used to relate the diameter of $G$ with the size of $G$, we obtain an analogous relation between the diameter of $H$ and the size of $H$ (Section \ref{sec:diamH1}, Proposition \ref{thm:dhnh2}), that can be expressed as $d_H=2^{O(\sqrt{\log n_H})}$. 

\item We improve the best upper bound known on $deg(H)$ (Section \ref{sec:last}, Theorem \ref{thm:degH}). We show this crucial result by using a different approach than the one used in the literature. We consider a node $u \in V(H)$ minimising the sum of distances and, instead of lower bounding the distance between two shopping vertices,  we introduce and study a natural kind of subsets, the $A$ sets (Section \ref{sec:1}). Each $A$ set corresponds to a node $v \in V(H)$ and a pair of edges $e_1,e_2$ where $v \in V(H)$ and $e_1,e_2 \in E(H)$ are two links bought by $v$. The $A$ sets play an important role when upper bounding the cost difference of player $v$ associated to the deviation of the same player that consists in selling $e_1,e_2$ and buying a link to $u$ (Section \ref{sec:1}, Proposition \ref{prop:formula1} and Proposition \ref{prop:formula2}). By counting the cardinality of these $A$ sets we show that the term $deg(H)$ can be upper bounded by an expression in which the terms $n,\alpha,n_H,$ and $d_H$  appear (Section \ref{sec:1}, Proposition \ref{thm:1}). By using the  relation $d_H=2^{O(\sqrt{\log n_H})}$ we can refine the upper bound for the $deg(H)$ even more. Subsequently, we consider the technique used in \cite{Mihalakmostly,Mihalaktree,Alvarezetal}, in which lower and upper bounds on the average degree of $H$ are combined to reach a contradiction whenever $H$ exists, i.e. whenever $G$ is a non-tree \NE graph.


\end{itemize}

\section{An upper bound for $deg(H)$ in terms of the size and the diameter of $H$}
\label{sec:1}

Remind that in all the sections we consider that $G$ is a \NE of a network creation game $\Gamma = (V, \alpha)$ where $\alpha > n$. If $G$ is not a tree then we denote by $H$ a maximal biconnected component of $G$. 

In this section we give an intermediate upper bound for the term $deg(H)$ that will be useful later to derive the main conclusion of this paper.

Let $u \in V(H)$ be a prefixed node and suppose that we are given $v \in V(H)$ and $e_1=(v,v_1),e_2=(v,v_2)$ two links bought by $v$. The \emph{$A$ set of $v,e_1=(v,v_1),e_2=(v,v_2)$}, noted as $A_{e_1,e_2}(v)$, is the subset of nodes $z \in V(G)$ such that every shortest path (in $G$) starting from $z$ and reaching $u$ goes through $v$ and the predecessor of $v$ in any such path is either $v_1$ or $v_2$.

Therefore, notice that $v \not \in A_{e_1,e_2}(v)$ and the following remark always hold: 

\begin{remark}
\label{rem:1}
Let $e_1,e_2,e_1',e_2'$ be four distinct edges such that $e_1,e_2$ are bought by $v$ and $e_1',e_2'$ are bought by $v'$. If $d_G(u,v)=d_G(u,v')$ then the $A$ set of $v,e_1,e_2$ and the $A$ set of $v',e_1',e_2'$ are disjoint even if $v=v'$.
\end{remark}

Notice that the definition of the $A$ sets depends on $u \in V(H)$, a prefixed node. For the sake of simplicity we do not include $u$ in the notation of the $A$ sets. Proposition 1 and Proposition 2 are stated for any general $u\in V(H)$ but in Corollary 1 we impose that $u$ minimises the function $D_G(\cdot)$ in $H$.

For any $i =1,2$, we define the \emph{$A^i$ set of $v,e_1=(v,v_1),e_2=(v,v_2)$}, noted as $A^i_{e_1,e_2}(v)$, the subset of nodes $z$ from $A_{e_1,e_2}(v)$ for which there exists a shortest path (in $G$) starting from $z$ and reaching $u$ such that goes through $v$ and the predecessor of $v$ in such path is $v_i$. 

With these definitions, $A_{e_1,e_2}(v) = A^1_{e_1,e_2}(v) \cup A^2_{e_1,e_2}(v)$ and $A^i_{e_1,e_2}(v)= \emptyset$ iff $d_G(u,v_i) = d_G(u,v)-1$ or $d_G(u,v_i)=d_G(u,v)$.  Furthermore, the subgraph induced by $A^i_{e_1,e_2}(v)$ is connected whenever $A^i_{e_1,e_2}(v) \neq \emptyset$.

Now, suppose that $e_1,e_2 \in E(H)$ and think about the deviation of $v$ that consists in deleting $e_i$ for $i=1,2$ and buying a link to $u$. Let $\Delta C$ be the corresponding cost difference and define $crossings(X,Y)$ for subsets of nodes $X,Y \subseteq V(G)$ to be the set of edges $xy$ with $x\in X$, $y \in Y$. Then we derive formulae to upper bound $\Delta C$ in the two only possible complementary cases: (i) $crossings(A^1_{e_1,e_2}(v),A^2_{e_1,e_2}(v)) \neq \emptyset$ and (ii) $crossings(A^1_{e_1,e_2}(v),A^2_{e_1,e_2}(v)) = \emptyset$. 

In case (i), $A^1_{e_1,e_2}(v),A^2_{e_1,e_2}(v) \neq \emptyset$ so that the subgraphs induced by $A^1_{e_1,e_2}(v),A^2_{e_1,e_2}(v)$ are both connected. This trivially implies that the graph induced by $A_{e_1,e_2}(v) = A^1_{e_1,e_2}(v) \cup A^2_{e_1,e_2}(v)$ is connected as well. Therefore, since $H$ is biconnected and $e_1,e_2 \in E(H)$ by hypothesis, there must exist at least one connection distinct from $e_1,e_2$ joining $A_{e_1,e_2}(v)$ with its complement. Taking this fact into the account we obtain the following result:

\begin{proposition}
\label{prop:formula1} Let us assume that  $crossings(A^1_{e_1,e_2}(v),A^2_{e_1,e_2}(v)) \neq \emptyset$ and $xy$ is any connection distinct from $e_1,e_2$ between $A_{e_1,e_2}(v)$ and its complement, with $x \in A_{e_1,e_2}(v)$. Furthermore, let $l$ be the distance between $v_1,v_2$ in the subgraph induced by $A_{e_1,e_2}(v)$. Then $\Delta C$, the cost difference for player $v$ associated to the deviation of the same player that consists in deleting $e_1,e_2$ and buying a link to $u$, satisfies the following inequality:

$$ \Delta C \leq -\alpha + n+D_G(u)-D_G(v)+(2d_G(v,x)+l)|A_{e_1,e_2}(v)|$$ 
\end{proposition}

\begin{proof}

The term $-\alpha$ is clear because we are deleting the two edges $e_1,e_2$ and buying a link to $u$. Now let us analyse the difference of the sum of distances in the deviated graph $G'$ vs the original graph. For this purpose, suppose wlog that $x \in A^1_{e_1,e_2}(v)$ and let $z$ be any node from $G$. We distinguish two cases:

(A) If $z \not \in A_{e_1,e_2}(v)$ then:

(1) Starting at $v$, follow the connection  $vu$.

(2) Follow a shortest path from $u$ to $z$ in the original graph.

In this case we have that: 

$$d_{G'}(v,z) \leq 1+d_G(u,z)$$

(B) If $z \in A_{e_1,e_2}(v)$ then there exists some $i$ such that $z \in A^i_{e_1,e_2}(v)$. Consider the following path (see the figure below for clarifications):

\begin{figure}
\begin{center}
\includegraphics[scale=0.5]{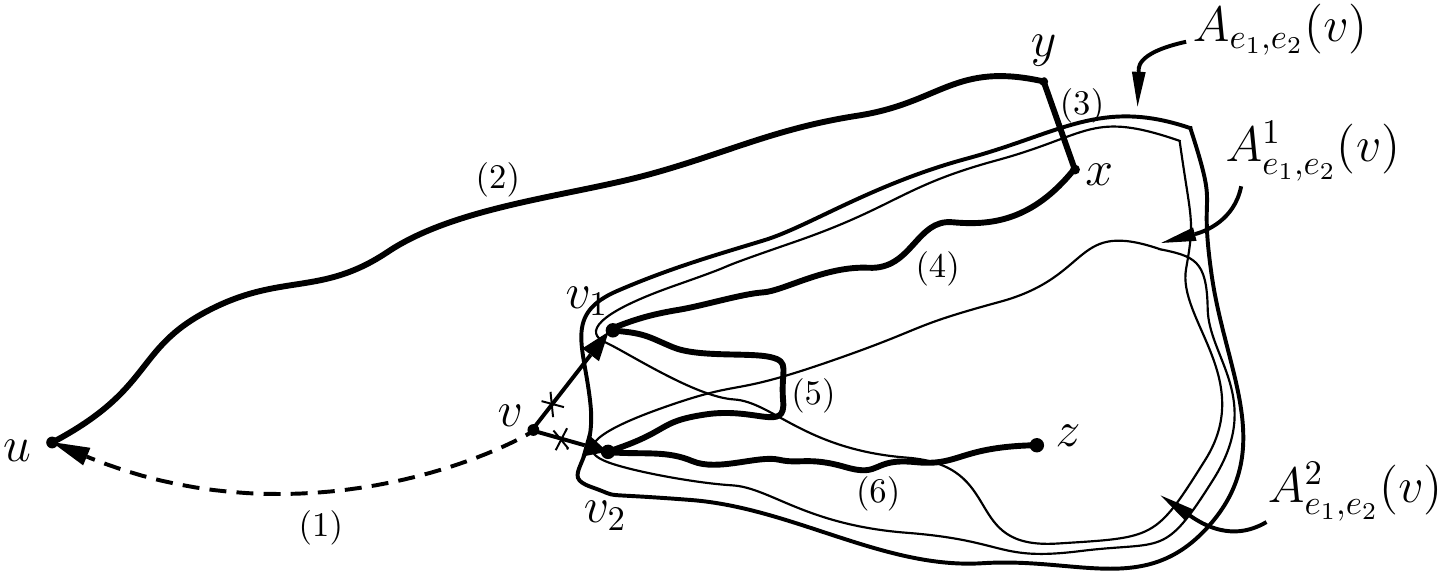}
\end{center}
\caption{The new path from $z$ to $v$ in the deviated graph $G'$}
\end{figure}
 (1) Starting at $v$, follow the connection $vu$, which corresponds to one unit distance.

 (2) Follow a path from $u$ to $y$ contained in the complementary of $A_{e_1,e_2}(v)$. Since $y \not \in A_{e_1,e_2}(v)$ we have that $d_G(u,y)\leq d_G(u,v)+d_G(v,x)+1$. Therefore, in this case we count at most $d_G(u,v)+d_G(v,x)+1$ unit distances. 
		
(3) Cross the connection  $yx$, which corresponds to one unit distance.

(4) Go from $x$ to $v_1$ inside $A_{e_1,e_2}(v)$ giving exactly $d_G(x,v)-1$ unit distances.

(5) Go from $v_1$ to $v_i$ inside $A_{e_1,e_2}(v)$ giving at most $l$ unit distances. 

(6) Go from $v_i$ to $z$ inside $A_{e_1,e_2}(v)$ giving exactly $d_G(v,z)-1$ unit distances.

In this case we have that:
\begin{align*}
d_{G'}(v,z) \leq & \overbrace{1}^{(1)} + \overbrace{d_G(u,v)+d_G(v,x)+1}^{(2)}+\overbrace{1}^{(3)}+\overbrace{d_G(x,v)-1}^{(4)}+\overbrace{l}^{(5)}+\overbrace{d_G(v,z)-1}^{(6)}  \\
 =& 1+d_G(u,z)+(2d_G(v,x)+l) 
\end{align*}
Combining the two inequalities we reach the conclusion: 
$$\Delta C \leq -\alpha + \sum_{z \in V(G)}( d_{G'}(v,z)-d_G(v,z)) \leq -\alpha +n+D_G(u)-D_G(v)+(2d_G(v,x)+l)|A_{e_1,e_2}(v)|$$
\end{proof} 
\QEDB

In case (ii), we assume that $crossings(A^1_{e_1,e_2}(v),A^2_{e_1,e_2}(v))=\emptyset$. Since $H$ is biconnected and $e_1,e_2 \in E(H)$ by hypothesis, for each $i$ such that $A^i_{e_1,e_2}(v) \neq \emptyset$ there must exist at least one connection distinct from $e_i$ joining $A_{e_1,e_2}^i(v)$ with its complement. Taking this fact into the account we obtain the following result:

 \begin{proposition}
 \label{prop:formula2} Let us assume that $crossings(A^1_{e_1,e_2}(v),A^2_{e_1,e_2}(v)) = \emptyset$ and let $I \subseteq \left\{1,2\right\}$ be the subset of indices $i$ for which $A^i_{e_1,e_2}(v) \neq \emptyset$. Furthermore, suppose that for each $i\in I$, $x_iy_i$ is any connection distinct from $e_i$ between $A^i_{e_1,e_2}(v)$ and its complement, with $x_i \in A^i_{e_1,e_2}(v)$. Then $\Delta C$, the cost difference of player $v$ associated to the deviation of the same player that consists in deleting $e_1,e_2$ and buying a link to $u$, satisfies the following inequality:

$$ \Delta C \leq -\alpha + n+D_G(u)-D_G(v)+\max(0,2 \max_{i \in I} d_G(v,x_i))|A_{e_1,e_2}(v)|$$ 

\end{proposition}

\begin{proof}

The term $-\alpha$ is clear because we are deleting $e_1,e_2$ and buying a link to $u$. Now let us analyse the difference of the sum of distances in the deviated graph $G'$ vs the original graph. To this purpose, let $z$ be any node from $G$. We distinguish two cases:

(A) If $z \not \in A_{e_1,e_2}(v)$ then:

 (1) Starting at $v$, follow the connection  $vu$.

(2) Follow a shortest path from $u$ to $z$ in the original graph. 

In this case we have that: 

$$d_{G'}(v,z) \leq 1 + d_G(u,z) $$

(B) If $z \in A_{e_1,e_2}(v)$ then there exists some $i$ such that $z \in A^i_{e_1,e_2}(v)$. Consider the following path:

 (1) Starting at $v$, follow the connection $vu$, which corresponds to one unit distance.

 (2) Follow a path from $u$ to $y$ contained in the complementary of $A_{e_1,e_2}(v)$. Since $ y\not \in A_{e_1,e_2}(v)$ we have that $d_G(u,y) \leq d_G(u,v)+d_G(v,x_i)+1$. Therefore, in this case we count at most $d_G(u,v)+d_G(v,x_i)+1$ unit distances. 
 		
 (3) Cross the connection $y_ix_i$, which corresponds to one unit distance. 

 (4) Go from $x_i$ to $v_i$ giving exactly  $d_G(x_i,v)-1$ unit distances.

 (5) Go from $v_i$ to $z$ giving exactly $d_G(v,z)-1$ unit distances.

In this case we have that:
\begin{align*}
d_{G'}(v,z) \leq & \overbrace{1}^{(1)} + \overbrace{d_G(u,v)+d_G(v,x)+1}^{(2)}+\overbrace{1}^{(3)}+\overbrace{d_G(x_i,v)-1}^{(4)}+\overbrace{d_G(v,z)-1}^{(5)}  \\
 =& 1+d_G(u,z)+2d_G(v,x_i) 
\end{align*}
Combining the two inequalities we reach the conclusion: 
$$\Delta C \leq -\alpha + \sum_{z \in V(G)}( d_{G'}(v,z)-d_G(v,z)) \leq -\alpha + n+D_G(u)-D_G(v)+\max(0,2 \max_{i \in I} d_G(v,x_i))|A_{e_1,e_2}(v)|$$

\end{proof} 
\QEDB

 Now, notice the following simple fact: 

\begin{remark}
\label{rem:distH}
If $z_1,z_2 \in V(H)$ then any shortest path from $z_1$ to $z_2$ is contained in $H$. This is because otherwise, using the definition of cut vertex, any such path would visit two times the same cut vertex thus contradicting the definition of shortest path. Therefore, if $z_1,z_2\in V(H)$ then $d_G(z_1,z_2) = d_H(z_1,z_2 ) \leq d_H$. 
\end{remark}

Combining the formulae from Proposition \ref{prop:formula1} and Proposition \ref{prop:formula2} together with this last remark, we can obtain a lower bound for the cardinality of any $A$ set of $v,e_1,e_2$ when $u$ satisfies a very specific constraint:

\begin{corollary}
 \label{corol:formula}  If $u \in V(H)$ is such that $D_G(u) = \min_{z \in V(H)} \left\{D_G(z) \right\}$, then $|A_{e_1,e_2}(v)| \geq \frac{\alpha -n}{4d_H}$

\end{corollary}

\begin{proof}
Let us analyse the properties that are fulfilled for the distinct elements in this setting:  

First, $u$ minimises the sum of distances on $V(H)$. Therefore, $D_G(u)-D_G(v)\leq 0$ for any $v\in V(H)$.

Now, let $xy$ be any crossing between $A_{e_1,e_2}(v)$ and its complement with $x$ in $A_{e_1,e_2}(v)$ and $y$ in the complementary of $A_{e_1,e_2}(v)$. Consider also $x',y'$ be the nodes from $V(H)$ such that $x\in S(x')$ and $y\in S(y')$. If $z\in V(H)$, then either $S(z)$ is a subset of $A_{e_1,e_2}(v)$, if $z \in A_{e_1,e_2}(v)$, or $S(z)$ is a subset of the complementary of $A_{e_1,e_2}(v)$ otherwise, by the definition of the $A$ sets and by the definition of cut vertex. Therefore, $S(x')$ is a subset of  $A_{e_1,e_2}(v)$ and $S(y')$ is a subset of the complementary of  $A_{e_1,e_2}(v)$. Furthermore, by the definition of biconnected component, any crossing or connection between $S(z_1)$ and $S(z_2)$ with $z_1,z_2 \in V(H)$ and $z_1 \neq z_2$, if it exists, must definitely be $z_1z_2$. Therefore, $x=x',y=y'$ and as a result $x,y\in V(H)$. Then by Remark \ref{rem:distH}, the distance from $x$ to $v$ is at most $d_H$. In a similar way, it can be deduced that if $x_iy_i$  is any crossing between $A_{e_1,e_2}^i(v)$ and its complement, then $x_i,y_i \in V(H)$ and therefore, the distance from $x_i$ to $v$ is at most $d_H$. As a conclusion, both expressions $d_G(v,x)$ and $d_G(v,x_i)$, appearing in the formulae from Proposition \ref{prop:formula1} and Proposition \ref{prop:formula2}, respectively, are at most $d_H$.

Moreover, whenever $e_1,e_2 \in E(H)$ and $crossings(A^1_{e_1,e_2}(v),A^2_{e_1,e_2}(v)) \neq \emptyset$, any shortest path connecting $v_1$ and $v_2$ inside $A_{e_1,e_2}(v)$ is contained in $H$ and has length at most $2d_H$. This implies that the expression $l$ appearing in the formula from Proposition \ref{prop:formula1} is at most $2d_H$.

With all these results, we deduce that, the expressions multiplying $|A_{e_1,e_2}(v)|$ in the rightmost term of the two inequalities from Proposition \ref{prop:formula1} ($2d_G(v,x)+l$) and Proposition \ref{prop:formula2} ($\max(0,2 \max_{i \in I}d_G(v,x_i))$)  can be upper bounded by $4d_H$. 

Imposing that $G$ is a \NE then we obtain the conclusion.
\end{proof}
\QEDB

Now we use this last formula to give an upper bound for the average degree of $H$. Recall that we are working in the range $\alpha > n$:

\begin{proposition}
\label{thm:1} 
$$deg(H) \leq 2+\frac{16d_H(d_H+1)n}{n_H(\alpha-n)}$$
\end{proposition}

\begin{proof}
For any node $v \in V(H)$ let $Z(v)$ be any maximal set of distinct and mutually disjoint pairs of edges from $H$ bought by $v$. Let $X$ be defined as the set of tuples $( \left\{ e_1,e_2 \right\} ,v)$ with $v \in V(H)$ and $\left\{ e_1,e_2 \right\}$ a pair of edges from $Z(v)$. Now define $S = \sum_{(\left\{ e_1,e_2 \right\},v) \in X} |A_{e_1,e_2}(v)|$. On the one hand, using Corollary \ref{corol:formula}: 
$$S \geq \frac{\alpha -n }{4d_H}|X|$$

On the other hand, for each distance index $i$, let $S_i$ be the sum of the cardinalities of the $A$ sets for all the tuples $(\left\{ e_1,e_2 \right\},v)\in X$ with $d_G(u,v)=i$. By Remark \ref{rem:1}, $S_i \leq n$. Therefore: 

$$|X| \frac{\alpha - n}{4d_H} \leq S = S_0+...+S_{d_H} \leq n(d_H+1)$$

Next, notice that there are exactly $\lfloor \frac{deg^+_H(v)}{2}\rfloor$ pairs in $Z(v)$ for each $v$ considered. Furthermore, $\lfloor \frac{deg_H^+(v)}{2}\rfloor = deg^+_H(v)/2$ if $deg_H^+(v)$ is even and $\lfloor \frac{deg_H^+(v)}{2}\rfloor =(deg^+_H(v)-1)/2$ otherwise. Hence:

$$|X| \geq \sum_{v \in V(H)}\frac{deg_H^+(v)-1}{2} = \frac{|E(H)|-|V(H)|}{2}$$

Finally: 

$$ deg(H) = \frac{2|E(H)|}{|V(H)|} \leq 2+ \frac{4 |X|}{|V(H)|} \leq 2+\frac{16(d_H+1)nd_H}{n_H(\alpha-n)}$$

\end{proof}
\QEDB




\section{The diameter of $H$ vs the number of nodes of $H$}
\label{sec:diamH1}

 In this section we establish a relationship between the diameter and the number of the vertices of $H$ which allows us to refine the upper bound for the term $deg(H)$ using the main result of the previous subsection. 

We start extending the technique introduced by Demaine et al in \cite{Demaineetal:07}. Instead of reasoning in a general $G$, we focus our attention to the nodes from $H$ reaching an analogous result. Since for $\alpha > 4n-13$ every \NE is a tree it is enough if we study the case $\alpha < 4n$.

For any integer value $k$ and $u \in V(H)$ we define $N_{k,H}(u) = \left\{ v \in V(H) \mid d_G(u,v) \leq k \right\}$, the set of nodes from $V(H)$ at distance at most $k$ from $u$.  With this definition in mind then $S_k(u) = \cup_{v \in N_{k,H}(u)} S(v)$ is the set of all nodes inside $S(v)$ for all $v\in V(H)$ at distance at most $k$ from $u$. In other words, $S_k(u)$ is the set of all nodes $z$ such that the first cut vertex that one finds when following any shortest path from $z$ to $u$ is at distance at most $k$ from $u$.

Furthermore, for any integer $k$ we define $m_k = \min_{u \in V(H)} |N_{k,H}(u)|$. That is, $m_k$ is the minimum cardinality that any $k$-neighbourhood in $H$ can have. 

\begin{lemma}\label{lem:ball}
Let $H$ be a biconnected component of $G$. For any integer $k\geq 0$, either there exists a node $u \in V(H)$ such that $|S_{4k+1}(u)| > n/2$ or, otherwise,  $m_{5k+1} \geq m_k k/4$. 
\end{lemma}

\begin{proof}
If there is a vertex $u\in V(H)$ with $|S_{4k+1}(u)| > n/2$, then the claim is obvious. Otherwise, for every vertex $u \in V(H)$, $|S_{4k+1}(u)| \leq n/2$. Let $u$ be any node from $V(H)$ minimising the cardinality of the balls of radius $5k+1$ intersected with $V(H)$. That is, $u$ is any node from $V(H)$ with $|N_{5k+1,H}(u)| = m_{5k+1}$. Let $Z=\left\{v_1,...,v_l\right\}$ be any maximal set of nodes from $V(H)$ at distance $4k+1$ from $u$ (in $H$) with the property that every two distinct nodes $v_i,v_j \in Z$, we have that $d_G(v_i,v_j) \geq 2k+1$ (see the left picture from Figure 2 for a visual clarification). 

Now, consider the deviation of $u$ that consists in buying the links to every node from $Z$ and let $G'$ be the new graph resulting from such deviation. Let  $z \in S(w)$ with $w \in V(H)$ and $d_G(w,u) \geq 4k+1$ and consider any shortest path (in $H$) from $w$ to $u$. Let $w_{\pi}$ be the node from any such shortest path at distance $4k+1$ from $u$. By the maximality of $Z$ there exists at least one node $v_w \in Z$ for which $d_G(v_w,w_{\pi}) \leq 2k$. The original distance between $z$ and $u$ is $d_G(z,u)=d_G(z,w)+d_G(w,u)$. In contrast, the distance between $z$ and $u$ in $G'$ satisfies the following inequality (see the right picture from Figure 2 for a visual clarification): 
$$d_{G'}(z,u) \leq 1+d_G(v_{w},w_{\pi})+d_G(w_{\pi},w)+d_G(w,z)  $$
$$\leq 1+2k+(d_G(u,w)-(4k+1))+d_G(w,z)=-2k+d_G(u,w)+d_G(w,z)$$

\begin{figure}
\label{fig:growingball}
\begin{center}
\includegraphics[scale=0.4]{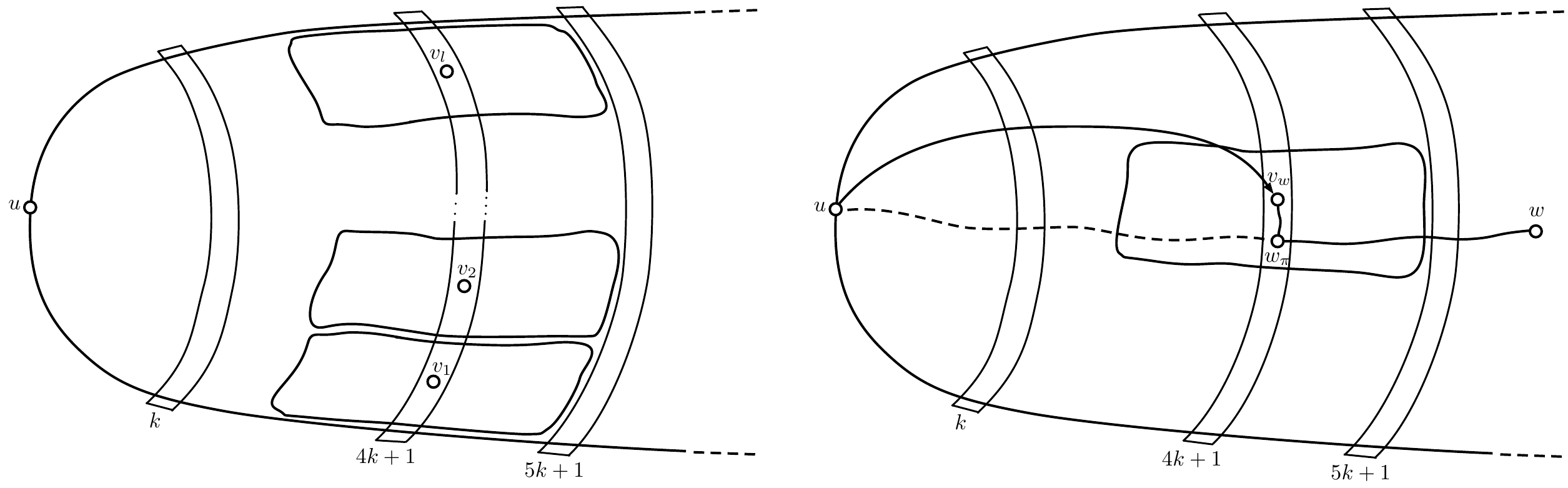}
\end{center}
\caption{The setting of nodes from the proof (left) and the alternative path from $w$ to $u$ in the deviated graph (right)}
\end{figure}

Therefore, $d_G(z,u)-d_{G'}(z,u) \geq 2k$. Since we are assuming that $|S_{4k+1}(u)| \leq n/2$ then this means that $\sum_{ \left\{ v \in V(H) \mid d_G(v,u) > 4k+1 \right\} }|S(v)| \geq n/2$, that is, the sum of the weights of the nodes from $H$ at distance strictly greater than $4k+1$ from $u$ is greater than or equal $n/2$. Then $\Delta C$, the cost difference for $u$ associated to such deviation, satisfies: 

$$\Delta C \leq l\alpha-2k\left( \frac{n}{2} \right) \leq 4nl- kn$$

Since $G$ is a \NE then from this we conclude that $l \geq k/4$. 
 
Finally, notice that the distance between two nodes in $Z$ is at least $2k+1$ implying that the set of all the balls of radius $k$ with centers at the nodes from $Z$ are mutually disjoint. Therefore, $m_{5k+1} =|N_{5k+1,H}(u)|  \geq l m_k \geq m_kk/4$. 

\end{proof}
\QEDB

\begin{lemma}\label{lem:reach}
If $r < d_H/4-4$ then $|S_{r}(u)| \leq n/2$ for every node $u\in V(H)$.
\end{lemma}

\begin{proof}
Suppose the contrary and we reach a contradiction, that is, suppose that there exists some $u\in V(H)$ with $|S_{r}(u)| > n/2$  and $r < d_H/4-4$. Let $t \in V(H)$ be any node at distance $d_H/2$ from $u$, which always exists. We consider the deviation of $t$ that consists in buying a link to $u$ and we define $G'$ to be the new graph resulting from such deviation. Let $z\in S_r(u)$ with $w\in V(H)$ such that $z \in S(w)$. The distance between $t$ and $w$ in $G$ is at least $d_H/2-r$  so the distance between $t$ and $z$ in $G$ is at least $d_H/2-r+d_G(w,z)$. In contrast, the distance between $t$ and $w$ in $G'$ is at most $1+r$, so the distance between $t$ and $z$ in $G'$ is at most $1+r+d_G(w,z)$. Therefore:
 
$$d_{G}(z,t)-d_{G'}(z,t) \geq d_H/2-2r-1 > d_H/2-2(d_H/4-4)-1 =7$$
Then $d_G(z,t)-d_{G'}(z,t) \geq 8$ and thus $\Delta C$, the cost difference of $t$ associated to such deviation, satisfies: 
$$\Delta C \leq \alpha -8|S_r(u)| \leq 4n-8|S_r(u)| < 4n-\frac{8}{2}n =0$$
A contradiction with the fact that $G$ is a \NE. 
\end{proof}
\QEDB

Combining these results we are able to give an extension of the result from Demaine et al in \cite{Demaineetal:07}:

\begin{proposition}
\label{thm:dhnh2} $d_H <  5^{\sqrt{2 \log_5 n_H}+5}$.

\end{proposition}



\begin{proof}

Consider the following sequence of numbers $(a_i)_{i \geq 0}$ defined in the following way: 

(i) $a_0=21$.

(ii) $a_{i+1} = 5a_i+1$ for $i \geq 0$. 

It is easy to check that $a_k = 21 \cdot 5^k + \frac{5^k-1}{5-1}$ for any $k \geq 0$ so that $22\cdot 5^k > a_k \geq  21 \cdot 5^k$. With this definition and using the two previous results we reach the conclusion that whenever $4a_i+1 < d_H/4-4$, then $|S_{4a_{i}+1}(u)| \leq n/2$ for all $u\in V(H)$, by Lemma \ref{lem:reach}, implying $m_{a_{i+1}} \geq m_{a_i}\frac{a_i}{4}$, by Lemma \ref{lem:ball}. Iterating the recurrence relation we can see that whenever $i \geq 0$ and $4a_i+1 < d_H/4-4$, then: 
$$m_{a_{i+1}} \geq  \frac{a_i a_{i-1}...a_1a_0}{4^{i+1}}m_{a_0}$$

Since $a_0=21$ then $m_{a_0} \geq 21$. Therefore: 

$$m_{a_{i+1}} \geq 21 \left(\frac{21}{4}\right)^{i+1} 5^{i+(i-1)+...+1+0} >5^{i^2/2}$$

Now, consider the value $k$ such that $4a_k+1 < d_H/4 -4\leq 4a_{k+1}+1$. On the one hand, $n_H \geq m _{a_{k+1}} >  5^{k^2/2}$ so this implies that $k\leq \sqrt{2 \log_5 n_H}$. On the other hand, $d_H/4 \leq 4a_{k+1}+5 < 22 \cdot 4 \cdot 5^{k+1}$. Therefore, $d_H < 5^{k+5}\leq 5^{\sqrt{2\log_5 n_H}+5}$, as we wanted to see.

\end{proof}
\QEDB

\section{The diameter of $G$ vs the diameter of $H$.}
\label{sec:diamH2}

In this section we establish a relationship between the diameter of $G$ and the diameter of $H$ when $\alpha > n$. Since for $\alpha > 4n-13$ every \NE is a tree it is enough if we study the case $n < \alpha < 4n$. 

We show that in this case, the distance between any pair $w,z \in V(G)$ where $z \in S(w)$, is upper bounded by $125$ from where we can conclude that $d_G < d_H+250$. To obtain these results we basically exploit the fact that $G$ is a \NE graph together with key topological properties of biconnected components: 


\begin{proposition}
\label{prop:diameter} Let $w \in V(H)$ and $z \in S(w)$ maximising the distance to $w$. Then $d_G(z,w) < 125$.
\end{proposition}
\begin{proof}
Let $Z$ be the subgraph of $G$ induced by $S(w)$ and $W$ the subgraph of $G$ induced by $w$ together with the set of nodes $V(G) \setminus S(w)$.  Then, define $r = d_G(z,w) = \max_{t \in V(Z)}d_G(w,t),s  = \max_{t \in V(W)}d_G(w,t)$ (see the figure below for clarifications). With these definitions it is enough to show that $r < 125$. Notice that, for instance, if $S(w) = \left\{ w \right\}$ then the result trivially holds. 

\begin{figure}[H]
\begin{center}
\includegraphics[scale=0.45]{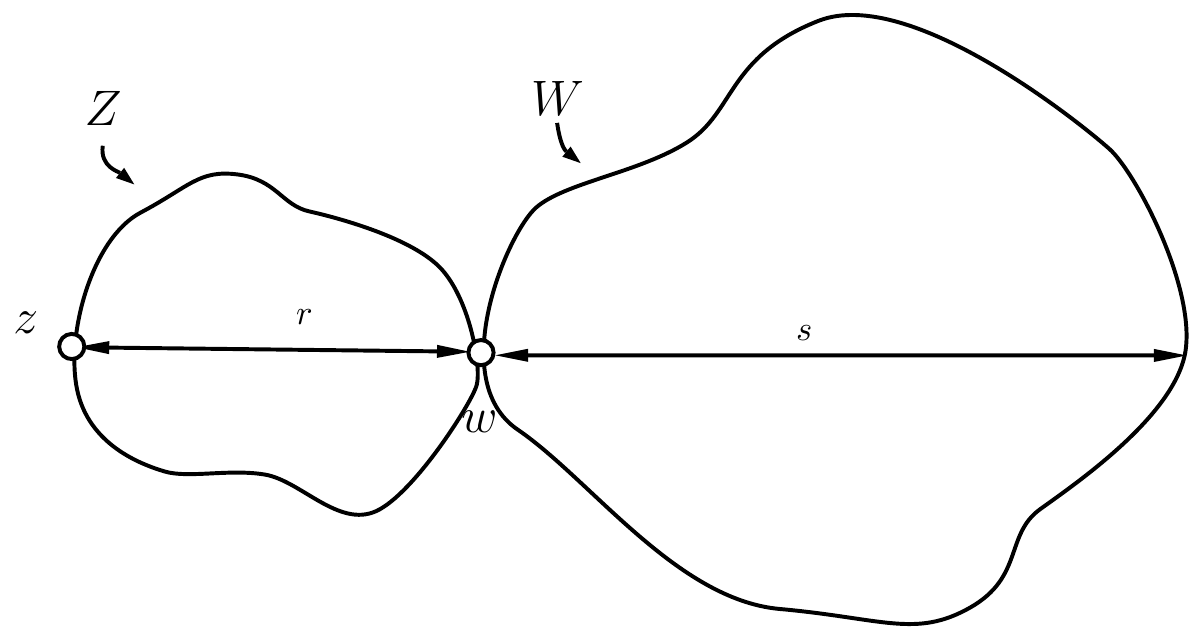}
\end{center}
\caption{The most important subsets, nodes and distances from the setting.}
\end{figure}

First, let us see that $\min(r,s) \leq 8$.

Let $v$ any node maximising the distance to $w$ in $W$ and $\Delta C_1$ and $\Delta C_2$ the corresponding cost differences of players $z$ and $v$, respectively, associated to the deviations of the same players that consist in buying a link to $w$. Then: 
$$\Delta C_1 \leq \alpha - |V(W)|(r-1) $$
$$ \Delta C_2 \leq \alpha - |V(Z)|(s-1)$$
Adding up the two inequalities and using that $\alpha < 4n$: 

$$ \Delta C_1+ \Delta C_2 \leq 2\alpha-(\min(r,s)-1)(|V(Z)|+|V(W)|) < 8n - (\min(r,s)-1)n$$

Since $G$ is a \NE graph then $\Delta C_1 + \Delta C_2 \geq 0$ and from here we deduce that $\min(r,s) \leq 8$, as we wanted to see.

\vskip 10pt

 If $r \leq 8$ then we are done. Therefore we must address the case $s \leq 8$. 
 
Next, since $H$ is a non-trivial biconnected component, there exist nodes $t,t'\in V(H)$ such that they are adjacent in $H$, $t$ has bought the link $e=(t,t')$ and one of the two following cases happen: either (i) $t$ is at distance $1$ from $w$, $t'$ is at distance $1$ or $2$ from $w$ or (ii) $t'$ is at distance $1$ from $w$ and $t$ at distance $2$ from $w$ (see the figure below for a clarification). 

\begin{figure}[H]
\begin{center}
\includegraphics[scale=0.45]{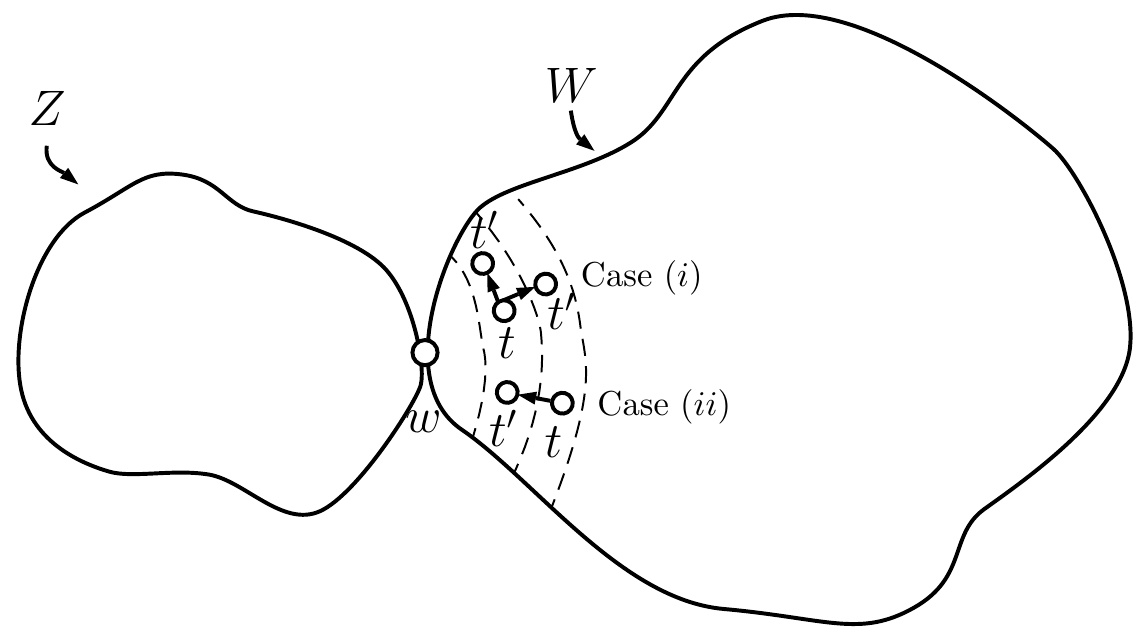}
\caption{An image depicting the setting for case (i) and case (ii).}

\end{center}
\end{figure}

In case $(i)$ we deduce that $|S(w)| = |V(Z)| \leq n \frac{4s-2}{4s-1} \leq n \frac{30}{31}$. This is because of the following reasoning. Let $\Delta C_{delete}$ be the corresponding cost difference of player $t$ associated to the deviation of the same player that consists in deleting the edge $e$. Since $H$ is biconnected then there exists a loop going through $e$ and contained in $H$ of length at most $4s+1$. Notice that when deleting $e$, $t$ only increases the distances maybe to the nodes from $V(W) \setminus \left\{w\right\}$ but not to the nodes from $V(Z)$ by at most $4s-1$ distance units. Therefore: 
$$\Delta C_{delete} \leq -\alpha + (4s-1)(n-|V(Z)|) < -n+(4s-1)(n - |V(Z)|)$$ 
Since $G$ is a \NE graph then $\Delta C_{delete} \geq 0$ and from here, using the hypothesis $s \leq 8$, we deduce the conclusion: 
$$|V(Z)| < \frac{-n+n(4s-1)}{4s-1} = n\frac{4s-2}{4s-1} \leq \frac{30}{31}n$$

\vskip 5pt

In case $(ii)$ we deduce that $|S(w)| = |V(Z)| \leq n/2$. This is because of the following reasoning. Let $\Delta C_{swap}$ be the corresponding cost difference  of player $t$ associated to the deviation of the same player that consists in swapping the edge $e$ for the link $(t,w)$. Notice that when performing such swap, $t$ only increases the distances maybe to the nodes from $V(W) \setminus \left\{w\right\}$ but strictly decreases for sure, one unit distance to all the nodes from $V(Z)$. Therefore: 

$$\Delta C_{swap} \leq -|V(Z)| + (n-|V(Z)|) \leq n -2 |V(Z)|$$ 

Since $G$ is a \NE graph then $\Delta C_{swap} \geq 0$ and from here we deduce the conclusion $|V(Z)| \leq n/2$.

Hence, we have obtained that either $|S(w)| \leq \frac{30}{31}n$, in case (i), or $|S(w)| \leq \frac{n}{2}$, in case (ii).

\vskip 5pt

Finally, consider the deviation of $z$ that consists in buying the link to $w$. Then the corresponding cost difference $\Delta C_{buy}$ satisfies the following inequality: 

$$\Delta C_{buy} \leq \alpha - (r-1)(n-|S(w)|) < 4n - (r-1)(n-|S(w)|)$$

Since $G$ is a \NE graph, then $\Delta C_{buy} \geq 0$ so that we conclude that $r < \frac{4n}{n-|S(w)|}+1$. Using this property we conclude that $r < 125$ in case (i) and $r \leq 8$ in case (ii), so we are done.

\end{proof}
\QEDB

As a consequence:

\begin{theorem}
\label{corol:diameter} $d_G < d_H+250$.
\end{theorem}

\section{Combining the results}
\label{sec:last}

Finally, in this section we combine the distinct results obtained so far to prove the main conclusion. 

On the one hand, combining Proposition \ref{thm:1} with Proposition \ref{thm:dhnh2} we reach the following result for the average degree of $H$:

\begin{theorem}
\label{thm:degH}
$$deg(H) < 2 + \frac{16n}{\alpha - n}\frac{5^{2\sqrt{2 \log_5 n_H}+10}}{n_H}$$

\end{theorem}

On the other hand, recall that from Lemma 4 and Lemma 2 from \cite{Mihalakmostly} and \cite{Mihalaktree}, respectively, the general lower bound $deg(H) \geq 2+ \frac{1}{16}$ that works for any $\alpha$ can be obtained.









With these results in mind we are now ready to prove the following strong result: 

\begin{theorem}
\label{thm:weakertreeconjecture}
Let $\epsilon > 0$ be any positive constant and $\alpha > n(1+\epsilon)$. There exists a constant $K_{\epsilon}$ such that every biconnected component $H$ from any non-tree Nash equilibrium $G$ has size at most $K_{\epsilon}$. 
\end{theorem}

\begin{proof}
Let $G$ be any non-tree \NE graph. Then there exists at least one biconnected component $H$. By Theorem \ref{thm:degH} when $\alpha > n(1+\epsilon)$ we have that $deg(H) < 2+ \frac{16}{\epsilon}\frac{5^{2\sqrt{2\log_5 n_H}+10}}{n_H}$. On the other hand, we know that for any $\alpha$, $deg(H) \geq 2+\frac{1}{16}$. Then this implies that there exists a constant $K_{\epsilon}$ upper bounding the size of $H$, otherwise we would obtain a contradiction comparing the asymptotic behaviour of the upper and lower bounds obtained for $deg(H)$ in terms of $n_H$. 
\end{proof}
\QEDB

In other words, the biconnected component conjecture holds for $\alpha > n(1+\epsilon)$.

Furthermore, recall that it is well-known that the diameter of any graph plus one unit is an upper bound for the \PoA and the \PoA for trees is constant. Therefore, we conclude that:

\begin{theorem}\label{thm:main} Let $\epsilon > 0$ be any positive constant. The price of anarchy is constant for $\alpha > n(1+\epsilon)$. 
\end{theorem}

\begin{proof}

Let $G$ be a \NE. If $G$ is a tree we are done, because the \PoA for trees is at most $5$. Therefore to prove the result consider that $G$ is a non-tree configuration. Then, $G$ has at least one non-trivial biconnected component $H$. On the one hand, by Theorem \ref{thm:weakertreeconjecture}, there exists a constant $K_{\epsilon}$ that upper bounds the size of  $H$. This implies that $d_H \leq n_H \leq K_{\epsilon}$. On the other hand, by Theorem \ref{corol:diameter}, $d_G \leq d_H+250$. In this way, $d_G \leq K_{\epsilon}+250$ and since $K_{\epsilon}+250$ is a constant, then the conclusion follows because the \PoA is upper bounded by the diameter plus one unit.

\end{proof}
\QEDB

\section{The conclusions}

The most relevant contribution we have made in this article is to show that the price of anarchy is constant for $\alpha > n(1+\epsilon)$. We have not been able to prove the tree conjecture for $\alpha > n$ by showing that there cannot exist any non-trivial biconnected component $H$ for the same range of $\alpha$. Instead, we have proved that for $\alpha > n(1+\epsilon)$, if $H$ exists, then it has a constant number of nodes. This property implies constant \PoA for the same range of $\alpha$. The technique we have used relies mostly on the improved upper bound on the term $deg(H)$ for $\alpha > n$. However, as in \cite{Mihalakmostly,Mihalaktree}, our refined upper bound still depends on the term $n/(\alpha-n)$, that tends to infinity when $\alpha$ approaches $n$ from above. This makes us think that either our technique can be improved even more to obtain the conclusion that the tree conjecture claims or it might be that there exist some non-tree equilibria when $\alpha$ approaches $n$ from above.

\end{document}